\documentclass[referee,natbib,runningheads]{svjour3}
\journalname{arXiv.org}
\smartqed  
\usepackage{graphicx}

\usepackage{balance}
\usepackage{xcolor}
\usepackage{booktabs}
\usepackage{amsmath}
\usepackage{graphicx}
\usepackage{hyphenat}
\usepackage[margin=3cm]{geometry}

\begin{document}
\title{On the maximization of likelihoods belonging to the exponential family using ideas related to the Levenberg-Marquardt approach}


\titlerunning{computing mle with a Levenberg-Marquardt-like approach}        

\author{Marco Giordan \and Federico Vaggi \and Ron Wehrens}


\institute{M. Giordan \at
              IASMA Research and Innovation Centre, Biostatistics and Data Managment\\
              \email{marco.giordan@fmach.it}           
           \and
           F. Vaggi \at
              IASMA Research and Innovation Centre, Integrative Genomic Research \\
              \email{federico.vaggi@fmach.it}                 
              \and 
              Ron Wehrens \at
              IASMA Research and Innovation Centre, Biostatistics and Data Managment \\
              \email{ron.wehrens@fmach.it}      
}

\date{Received: date / Revised: date}

\maketitle

\begin{abstract} 
The Levenberg-Marquardt algorithm is a flexible iterative procedure used to solve non-linear least squares problems. In this work we study how a class of possible adaptations of this procedure can be used to solve maximum likelihood problems when the underlying distributions are in the exponential family. We formally demonstrate a local convergence property and we discuss a possible implementation of the penalization involved in this class of algorithms. Applications to real and simulated compositional data show the stability and efficiency of this approach.

\keywords{Aitchison distribution \and compositional data \and Dirichlet distribution \and generalized linear models \and natural link \and optimization.}
\end{abstract}

\section{Introduction}

The exponential family is a large class of probability distributions widely used in statistics. Distributions in this class possess many useful properties that make them useful for inferential and algorithmic purposes, especially with reference to maximum likelihood estimation. Despite this, when closed form solutions are not available, computational problems can arise in the fitting of a particular distribution in this family to real data. Starting values far from the optimum and/or bounded parameter spaces may cause common optimization algorithms like Newton-Raphson to fail. Convergence along a canyon and the phenomenon of \textit{parameter evaporation} (i.e. the algorithm is lost in a plateau and pushes the parameters to infinity) can be problematic too \citep[see][]{Transtum2012}. In such cases algorithms with more reliable convergence must be employed to find the maximum likelihood estimates.

The Levenberg-Marquardt algorithm was developed for the minimization of functions expressed as the sum of squared functions, usually squared errors, where its convergence properties have been theoretically demonstrated. Its performance is usually very good for functions that are mildly non-linear, and therefore many authors have attempted to adapt it to the maximization of the likelihood. In literature such adaptations have been based upon good heuristics and have been shown to provide reliable results, but to the best of our knowledge formal arguments for their convergence are still lacking. \citet{Smyth2002} gave an algorithm for REML (restricted or residual maximum likelihood) scoring with a Levenberg-Marquardt restricted step. It was applied to find estimates in heteroscedastic linear models with data normally distributed, and, taking advantage of the strong global convergence properties of the Levenberg-Marquardt algorithm, the author concluded that the result \textit{\ldots can therefore be expected to be globally convergent to a solution of the REML equations subject to fairly standard regularity conditions \ldots}. \citet{statmodSmyth} considered another adaptation of the Levenberg-Marquardt algorithm to fit a generalized linear model with secure convergence (for gamma generalized linear model with identity links or negative binomial generalized linear model with log-links, see the R package \textit{statmod}).  \citet{aitchison2003statistical} used a Levenberg-Marquardt step for the maximization of the likelihood of two distributions on the simplex that belong to the exponential family. \citet{Stinis2005} used the standard Levenberg-Marquardt algorithm to minimize an error function related to a moment-matching problem arising in maximum likelihood estimation of distributions in the exponential family.

In this paper we evaluate a class of adaptations of the Levenberg-Marquardt algorithm to find the maximum likelihood estimates in the exponential family and we give formal proof of convergence for such an algorithm in the general case of generalized linear models with natural link. We examine its performance on real and simulated data.

\section{Exponential Family}

We consider the problem of maximum likelihood estimation in the exponential family. With $y$ we denote a $K$-dimensional observation. The densities of the distributions in the exponential family can be written in their natural parametrization as
\begin{equation}\label{EF}
f(y|\theta) = \exp \left\{ y' \theta - b(\theta) + c(y) \right\}
\end{equation}
where the $K$-dimensional vector $\theta$ is the natural parameter, $b(\theta)$ is the log-partition function and $\exp\{c(y)\}$ is the base measure. The natural parameter space is the set of natural parameters for which the log-partition function is finite. The covariance matrix of a random vector with distribution in the exponential family is usually supposed to be positive definite in the interior of the natural parameter space.

The Newton-Raphson algorithm and the Fisher scoring algorithm are two algorithms that are commonly used to find maximum likelihood estimates. They are equivalent in the multivariate exponential family with natural parametrization because for a sample of $n$ i.i.d. observations the Hessian matrix of the loglikelihood is $n$ times the Hessian of the log-partition function and is independent of the data. Further, for densities in the exponential family the Hessian is related to the covariance matrix and it is invertible, insuring good convergence properties for the algorithms. In what follows we denote with $l$ the loglikelihood function, with $s$ the corresponding score function (the gradient of the loglikelihood function with respect to the parameters) and with $H_l$ the Hessian of the loglikelihood. With this notation a Newton-Raphson iteration can be expressed as
\begin{equation}\label{NRiteration}
\theta^{(t+1)} = \theta^{(t)} - [{H_l(\theta^{(t)})}]^{-1} s(\theta^{(t)}).
\end{equation}
The algorithm can also be written in the equivalent form of IRLS, Iteratively Reweighted Least Squares  (see appendix A). In appendix A we provide all the necessary notation and results to extend the Newton-Raphson algorithm to the case of generalized linear model with natural link. They will be the basis for the proof in appendix B.

\section{A class of algorithms for maximum likelihood estimation}\label{sec:LM}

The Levenberg-Marquardt algorithm \citep{Levenberg1944,Marquardt1963} is an adaptive algorithm that is used to solve least squares problems. It is based on a modification of the Gauss-Newton method. Specifically, the Levenberg-Marquardt algorithm uses in its iterations a penalized version of $J'J$ where $J$ is the Jacobian matrix of the target function, which guarantees that the matrix can be inverted. The penalization is opportunely tuned through a damping parameter. In practice the Hessian of the function to be minimized is replaced by
\begin{equation}\label{dampLM}
J'J + \gamma \mathrm{diag}(J'J)
\end{equation}
where $\gamma$ is the damping parameter. Here we show how similar ideas can be applied directly to the maximization of the loglikelihood in exponential families. In appendix B we give a formal proof of a local convergence property for the algorithm outlined below when this is applied to generalized linear model with natural link. This case includes the case of distributions in the exponential family as can easily be seen considering design matrices equal to identity matrices (see appendix A).

The algorithm for maximum likelihood estimation in the exponential family that we evaluate in this work can be thought as a penalized version of the Newton-Raphson algorithm, with a penalization similar to the one used in the Levenberg-Marquardt algorithm, \textit{i.e.}, Equation (\ref{dampLM}). The iterations are given by
\begin{equation}\label{LMiteration}
\theta^{(t+1)} = \theta^{(t)} - [{H_l(\theta^{(t)})} + \gamma^{(t)}{P(\theta^{(t)})}]^{-1} s(\theta^{(t)})
\end{equation}
where $\gamma^{(t)}$ is a positive damping parameter and $P(\theta^{(t)})$ is a symmetric negative definite matrix. In practice $P(\theta^{(t)})$ is always a diagonal matrix to be used as penalization. The proposal of \citet{Levenberg1944} for the least squares problem was to use as penalization the identity matrix $I$. For maximum likelihood estimation the corresponding proposal is $P(\theta^{(t)})=-I$. The proposal of \citet{Marquardt1963} instead corresponds to $P(\theta^{(t)})=\mathrm{diag} H_l{(\theta^{(t)})}$. In this paper we use the latter form which has the advantage of better following the curvature of the function being maximized. The damping parameter will play a central role to make the algorithm adaptive. This can decrease the penalty, and thus speed up the convergence. Moreover, for a bounded parameters space, a careful tuning of the damping parameter is required to avoid large steps that could bring the parameters outside the allowed region. The penalization $\gamma^{(t)}{P(\theta^{(t)})}$ can be used to ensure that $[{H_l(\theta^{(t)})} + \gamma^{(t)}{P(\theta^{(t)})}]$ is invertible. In regular exponential families the inversion of $H_l$ is usually possible, but such a matrix can still be poorly conditioned and therefore algorithms for matrix inversion can fail or perform poorly. The penalization can avoid this problem. Other features related to the implementation of the algorithm will be discussed in the next section.

The iterations in Equation (\ref{LMiteration}) are a sound basis for a stable algorithm. In appendix B we give a formal proof for the convergence of the algorithm.

\subsection{Damping Parameter and Stopping Criteria}\label{secDampPar}

The damping parameter $\gamma$ in Equation (\ref{LMiteration}) influences the step size of the iterations. As $\gamma$ reaches zero, the algorithm reduces to the Newton-Raphson algorithm. Since the Newton-Raphson algorithm has a quadratic rate of convergence in the neighbourhood of the maximum we want a small $\gamma$ in such a situation. However, if the current iteration is far from the maximum, small steps can avoid some of the problems discussed in the previous section. These steps are provided by a large value of $\gamma$. To achieve the desired changes in the damping parameter we adopt the strategy proposed by \cite{Nielsen1998}. We use the following gain function:
\begin{displaymath}
\varrho^{(t+1)} = \frac{l(\theta^{(t+1)})-l(\theta^{(t)})}{ -0.5(\theta^{(t+1)}-\theta^{(t)})'H_l(\theta^{(t)})(\theta^{(t+1)}-\theta^{(t)})}.
\end{displaymath}
With this function we are comparing the actual increase or decrease of the loglikelihood function with the second order term of its Taylor approximation. The denominator is always positive (see appendix A) and therefore a positive value of the gain function indicates that we are moving in the right direction. With a large positive value of the gain function we can reduce the parameter $\gamma^{(t+1)}$. In this way we approximate the Newton-Raphson algorithm. A small value of $\varrho^{(t+1)}$ indicates instead  that the Taylor's approximation is not working very well and in this case it is better to penalize the steps by increasing $\gamma^{(t+1)}$. 
The Hessian matrix in the denominator of the gain function can be replaced by its actual approximation in the squared brackets of Equation (\ref{LMiteration}). We do this when small values of the Hessian lead to values close to zero in the denominator of the gain function.
The damping parameter can be updated as follows: 
\begin{equation}\label{DampPar}
 \nonumber \gamma^{(t+1)} = \gamma^{(t)} \max \left\{\frac{1}{3},1 - (2 \varrho^{(t+1)} - 1)^3 \right\}I_{\varrho^{(t+1)} > 0} + \gamma^{(t)} 2  I_{\varrho^{(t+1)} \leq 0} 
\end{equation}
where $\gamma^{(0)}=1$. This value for $\gamma^{(0)}$ proposes therefore an initial penalized step. Other choices of course are possible.  The above function is similar to the one suggested by \citet{IMM2004}, it is positive and is continuous in $\varrho^{(t+1)}$. 

We stop the algorithm as soon as one of three criteria is reached: if the norm of the score is very close to zero: $ \parallel s(\theta^{(t)}) \parallel < \epsilon_1$; if the changes in the parameters in successive iterations are very small: $ \parallel \theta^{(t+1)} - \theta^{(t)} \parallel < \epsilon_2 ( \parallel \theta^{(t)} \parallel + \epsilon_2) $; if the number of iterations is greater than a pre-established threshold, \textrm{maxit}. In this work we consider $\epsilon_1=\epsilon_2=10^{-8}$ and \textrm{maxit}=1000. The algorithm is considered to have reached convergence only if the final estimates are inside the natural parameter space.

\section{Examples}

In this section we apply the algorithm described above to two distributions belonging to the exponential family. Specifically, we focus on two distributions for the analysis of compositional data, \textsl{i.e.}, positive data that sum up to one. 

Within numerical accuracy, all the algorithms we compared reach the same optimum upon convergence (the distributions in the natural exponential family are convex), therefore we focus our comparisons on the computational efficiency and stability.

\subsection{First case: Dirichlet distribution}

If we denote with $\alpha = (\alpha_1,\ldots,\alpha_K)$ the vector of parameters of the Dirichlet distribution then its loglikelihood for $n$ i.i.d. observations can be written as
\begin{equation}\label{lDirichlet}
l(\alpha)=n \ln \Gamma(\sum_{k=1}^K \alpha_k) - n\sum_{k=1}^K \ln \Gamma(\alpha_k) +  \sum_{k=1}^K (\alpha_k - 1)\sum_{r=1}^n \ln y_{rk}.
\end{equation}
With this notation the parameters must be greater than zero. The transformation $\theta_k = \alpha_k - 1$ gives the natural parameters. The Dirichlet distribution can arise as a transformation of independently distributed gamma variables and as a consequence has independence properties that can bound its use \citep[see][p. 59-60]{aitchison2003statistical}. For high dimensional data however it is reasonable to assume that many components are almost independent and therefore we study its performance in this situation.

We compare the algorithm in (\ref{LMiteration})  (henceforth, simply LM) with the Newton-Raphson (NR) algorithm and a fixed point iteration (FPI) algorithm. The former has quadratic convergence in a neighbourhood of the maximum but can often fail, whereas the latter is very stable but can be slow. We also consider an implementation of (\ref{LMiteration}) with a fixed damping parameter \citep[see][]{giordan2014} to evaluate the differences with the suggested adapting damping parameter.  As starting values for the algorithms we employ four different strategies: the method of moments and the proposals based upon the works of \citet{Dishon1980}, \citet{Ronning1989} and \citet{Wicker2008}. The method of moments has the advantage of simplicity but works only on the marginal distributions and can give estimates outside the natural parameter space. The method of \citet{Dishon1980} is an improvement that considers information from all the data for the estimation of each single parameter. The proposal of \citet{Ronning1989} always gives initial parameters inside the parameter space, whereas \citet{Wicker2008} developed an approximation of the likelihood that is useful for high-dimensional data.

\subsubsection{Simulations}

In this simulation study we generate data from a Dirichlet distribution with dimension 1000. The number of simulated samples is equal to 20.  To avoid the rounding of many randomly generated values to zero due to machine precision we simulate data from distributions with large parameters. Specifically we consider an hypothetic sum of the parameters, $\sum \alpha_k$, ranging from 10000 to 50000 with step size of 2000 and we draw each parameter from a uniform distribution between $\sum \alpha_k/K -2$ and $\sum \alpha_k/K +2$ with $K=1000$. 

In \figurename ~\ref{fig:Sim} the convergence rate for each combination of starting values and methods (upper panel) and the mean number of iterations required for convergence (lower panel) are shown. NR has convergence problems when the starting values are given by the methods of moments. The other two methods instead are very stable. They provide convergence with all the starting values strategies. The increasing number of iterations needed for FPI (lower panel) suggest that by raising \textrm{maxit} we can always have  convergence for this method. NR, as expected, shows a very fast convergence. FPI requires a large number of iterations for convergence. The LM implementation with the fixed damping parameter has a performance very close to that of FPI. The adapting damping parameter instead brings the required number of iterations for convergence very close to the number of iterations that NR is using, while having greater stability than the FPI algorithm.

\begin{figure*}[!p]
\centering
\includegraphics[scale=0.50,angle=270]{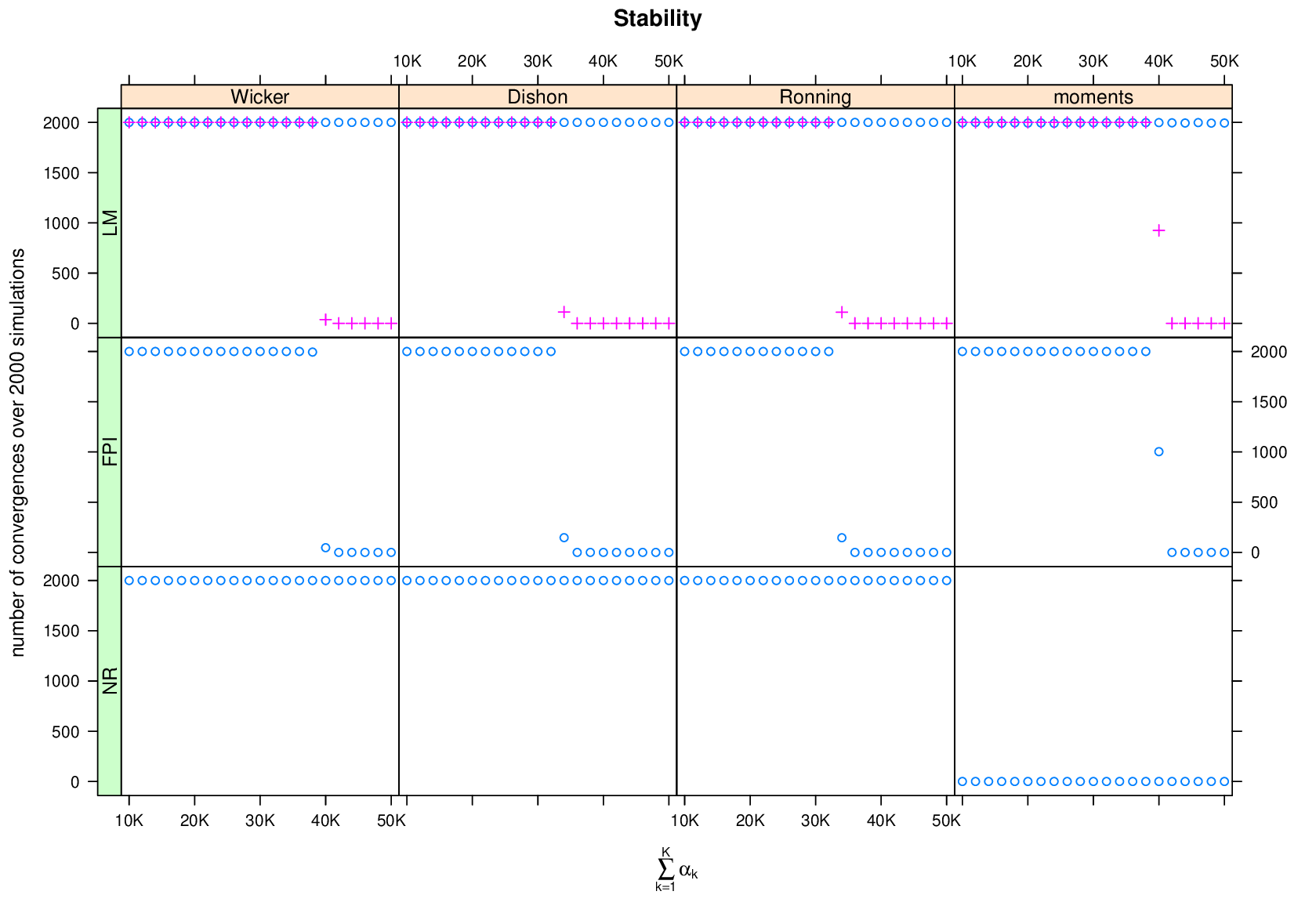}
\includegraphics[scale=0.50,angle=270]{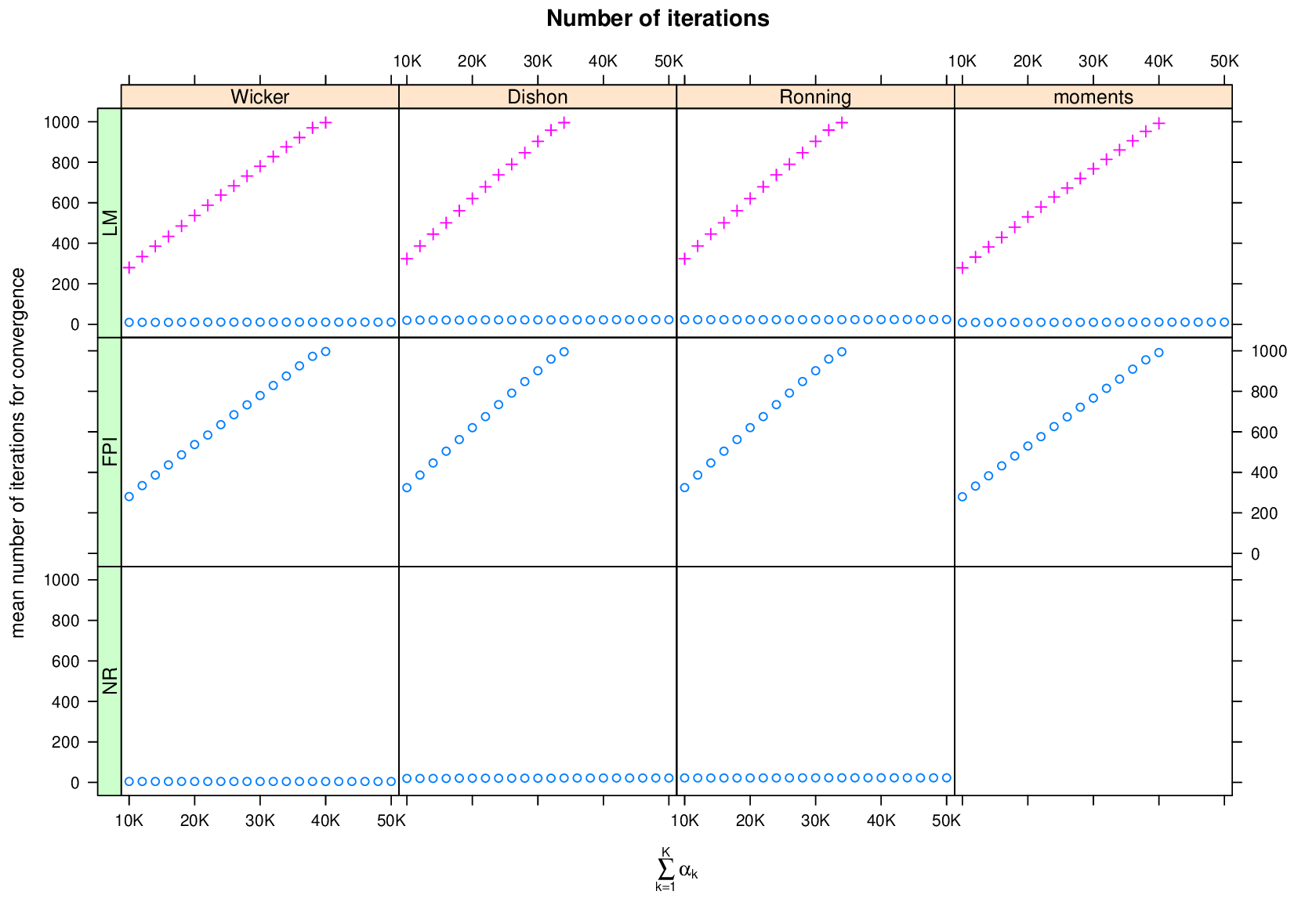}
\caption{simulations from the Dirichlet distribution. We compare three algorithms (LM, FPI and NR) and four different strategies for the starting values (Wicker, Dishon, Ronning and moments). With plus we report the results of LM using a fixed damping parameter.}\label{fig:Sim}
\end{figure*}

\subsubsection{Apple data set}

We now want to compare the performance of the algorithms on real data. We do this only for the two best algorithms from the previous simulation study: LM with the adapting damping parameter and NR. The apple data set \citep{Franceschi2012} contains mass spectrometric measurements and it is publicly available in the R package BioMark \citep{Wehrens2012}. We consider a subset of twenty samples (10 controls and 10 spiked-in samples) of positive ionization data. After an appropriate normalization to get compositional data the final data set has 1602 variables. The convergence results are reported in \tablename ~\ref{tab:Apple}.

The LM algorithm always reaches convergence for all the starting values strategies whereas NR fails to converge for two out of four initialization strategies. The number iterations required for reaching convergence is very similar to the number required by the NR algorithm when it does reach convergence. Therefore the LM variant is both stable and fast.

\begin{table*}
\centering
\caption{results for the apple data set. LM converges for all initailization methods, whereas NR fails to converge in two out of four cases.}\label{tab:Apple}
\begin{tabular}{lcccc}
\toprule 
 & convergence NR & convergence LM & n iterations NR & n iterations LM \\ 
\midrule 
Wicker & no & yes & - & 11 \\ 
Dishon & yes & yes & 21 & 22 \\ 
Ronning & yes & yes & 30 & 31 \\ 
moments & no & yes & - & 55 \\ 
\bottomrule
\end{tabular} 

\end{table*}

\subsection{Second Case: Aitchison distribution}

For low dimensional data the independence structure of the Dirichlet distribution can be too strong to allow a good fitting. \citet[][p. 310-315]{aitchison2003statistical} introduced a more flexible alternative that arises as a generalization of the Dirichlet distribution and the additive logistic normal distribution. We refer to it as the \textit{Aitchison distribution} in this paper. Its parametric form is again in the exponential family and the corresponding loglikelihood can be expressed as:
\begin{equation}\label{lAitchison}
l= -n \log c(\alpha,\beta) + \sum_{i=1}^{K}(\alpha_i - 1)U_i + \sum_{i=1}^{K-1}\sum_{j=i+1}^{K} \beta_{ij}V_{ij}
\end{equation}
where
$$
U_i=\sum_{r=1}^n \log y_{ri}  
$$
and
$$
V_{ij}= -\frac{1}{2}\sum_{r=1}^n (\log y_{ri} - \log y_{rj})^2  
$$
with $i=1,\ldots,K-1; j=i+1,\ldots,K$. We remark that the Aitchison distribution has a more complex structure than the additive logistic normal distribution. Contrary to the additive logistic normal and the normal on simplex where a change in coordinate allows to work with standard techniques for the multivariate normal distribution \citep[see, for example,][]{Figueras2008}, the Aitchison distribution requires computational tools to find the maximum likelihood estimates.

It is evident from equations (\ref{lDirichlet}) and (\ref{lAitchison}) that the price to pay for the generalization is the increased number of parameters. This can be quite high, even for a data set of modest dimension. For a compositional data set with $K$ variables the number of parameters to be estimated is:
$$
K + \frac{1}{2}\left( K(K-1) \right)= \frac{1}{2}\left( K(K+1) \right).
$$
Moreover, the normalizing factor $c(\alpha,\beta)$ has no closed form and therefore it and its derivatives must be evaluated numerically. The computation of the loglikelihood is therefore particularly demanding even for a low dimensional data set. For example, for a compositional data set with 5 variables the Dirichlet distribution requires only 5 parameters while the Aitchison distribution requires 15.  Further, the lack of a closed form for the Hessian matrix implies the numerical evaluation of 120 integrals in each iteration. 

For the Dirichlet distribution the natural parameter space is the set of vectors with positive elements $\{ \alpha =(\alpha_1,\ldots,\alpha_K) \ | \ \alpha_i > 0 \ \mathrm{for} \ i=1,\ldots,K \}$. For the Aitchison distribution a similar description of the natural parameter space is not available. In \citet[][p. 311-312]{aitchison2003statistical} two different restrictions are proposed to obtain proper density functions. However, these conditions are sufficient but not necessary. In practice the normalizing constant of a proper density is finite and for a current vector of parameters this must be numerically evaluated in the algorithms. Therefore if the algorithms converge to finite values these must be in the natural parameter space because the corresponding normalizing constant (log-partition function) must be finite. To accurately calculate the normalizing factor $c(\alpha,\beta)$ we use Gauss-Hermite integration following the suggestions in \citet[][p 314-315]{aitchison2003statistical}.

The closed-form solutions for the maximum likelihood estimates of the additive logistic normal distribution will be used as starting values for the algorithms. These estimates are simply the sample mean and variance of the additive log-ratio transformation of the original compositional data \citep[see][p. 113, 313-314]{aitchison2003statistical}.

\subsubsection{Simulations}

For the Aitchinson distribution, the calculation of the gradient requires the calculation of several complex numerical integrals, and we therefore only examined a case with a low number of variables.
To generate samples from the Aitchison distribution we used the R package \textit{compositions} version 1.30-1 \citep{van2013analyzing}. To ensure good starting values we used 2000 simulations where the sample covariance matrix from the log-ratio transformation was positive definite. The number of samples in each simulated data set was 20. The dimensions of the compositional data sets were 3 and 5, corresponding to parameter vectors ($\alpha,\beta$) of length 6 and 15 respectively. The parametrization used in the above package is slightly different from the one used in the paper; for the simulations we considered parameters according to the example of \citet[][p. 66]{van2013analyzing}. 

A summary of the simulations is given in \tablename ~\ref{AitSim}. The number of succesful convergences dramatically increases when we use LM instead of NR, and the increase in the number of iterations required (in the subset of cases when NR does converge) is very modest, making it an obvious improvement for maximum likelihood estimation.

\begin{table*}
\centering
\caption{simulations from the Aitchison distribution. The first column indicates the number of parameters to be estimated. In columns two and three the number of convergences over 2000 simulations are given for the NR algorithm and the LM algorithm, respectively. Similarly, the last two columns show the mean numbers of iterations in case of convergence.}\label{AitSim}
\begin{tabular}{ccccc}
\toprule
n parameters & n convergences NR & n convergences LM & n iterations NR & n iterations LM \\ 
\midrule
6 & 22 & 879 & 7.95 & 13.32 \\ 
15 & 12 & 745 & 8.75 & 12.90 \\ 
\bottomrule
\end{tabular} 
\end{table*}

\subsubsection{Applications to known data sets}

We now apply the NR and LM algorithms to four data sets and we compare their performance.  The data sets are publicly available in the R package \textit{robCompositions} \citep{Templ2011} and/or the package \textit{compositions}. They are briefly described below (more information is available inside the actual packages):
\begin{itemize}
\item[] Data set 1, skye lavas. It is a data set with 3 variables: magnesium, sodium-potassium and iron. We used the variable \textit{iron} for the log-ratio transform.
\item[] Data set 2, arctic lake. It is a data set with 3 variables: clay, silt and sand. We used the variable \textit{sand} for the log-ratio transform.
\item[] Data set 3, machine operators. It is a data set with 4 variables: high quality production, low quality production, setting and repair. We used the variable \textit{repair} for the log-ratio transform.
\item[] Data set 4, expenditures. It is a data set with 5 variables: housing, food stuffs, alcohol, services and other. We used the variable \textit{other} for the log-ratio transform.
\end{itemize}
In \tablename ~\ref{AitDataSets} we summarize the convergence results for the Aitchison distribution. Only LM is able to give convergence for the arctic lake data set. Both algorithms failed to converge for the expenditures data set. For the remaining data sets both algorithms converge to the same parameters although the number of required iterations by LM is slightly greater than those required by NR.

\begin{table*}
\centering
\caption{analysis of 4 data sets with the Aitchison distribution. The first column indicates the data sets. In the columns two and three we see if the NR algorithm and the LM algorithm have converged. In the last two columns we give the number of iterations used for the  convergence.}\label{AitDataSets}
\begin{tabular}{lcccc}
\toprule
 & convergence NR & convergence LM & n iterations NR & n iterations LM \\ 
\midrule
skye lavas & yes & yes & 5 & 14 \\ 
arctic lake & no & yes & - & 12 \\ 
machine operators & yes & yes & 5 & 14 \\ 
expenditures & no & no & - & - \\ 
\bottomrule
\end{tabular} 
\end{table*}

\section{Discussion  and conclusions}

In this paper we have investigated the use of the Levenberg-Marquardt algorithm to find the maximum likelihood estimates of distributions in the exponential families. We have given formal proof of convergence to the optimum for a class of possible adaptations and we have shown through real and simulated data that the LM-variant outperforms other algorithms in many settings.

The penalization used in the paper is related to the curvature of the loglikelihood and it ensures a well-conditioned negative definite matrix in each iteration of the algorithm. This provide a stable algorithm at the price of an increased number of iterations for the convergence (using as reference the Newton-Raphson algorithm). However, the damping parameter used in the penalization is adaptive and therefore it can speed up the convergence of the algorithm.  We have shown in a simulation study that the difference with a fixed damping parameter can be substantial. We have compared the efficiency of the algorithms taking into account the number of iterations rather than the time. Since Equation (\ref{LMiteration}) is essentially Equation (\ref{NRiteration}) with the added computation for the penalty and the Fixed Point Iteration algorithm requires a much greater computational effort than both other algorithms, we can guarantee that the order of the efficiency comparison is in any case preserved. 

In this work we have focused our attention on distributions for the analysis of compositional data. The computational performance improvements, however, are expected to hold also for other distributions in the exponential family because the algorithm is not related to compositional data. In particular we have used the Dirichlet distribution and the Aitchison distribution to analyze high and low dimensional data respectively. For both distributions the adaptation of the Levenberg-Marquardt algorithm has shown substantial stability advantages over other algorithms. Despite this, the number of required iterations is very low thanks to the adaptive damping parameter. The algorithm studied in the paper is therefore a powerful computational tool for maximum likelihood estimation in the exponential family.

\bibliography{Biblio.bib}
\bibliographystyle{plainnat}

\section*{Appendix A}

In this appendix we summarize the multivariate generalized linear models using notation and results of \citet{Fahrmeir2001} and we refer the reader to their book for the regularity conditions usually assumed. The $K$-dimensional random variable $Y$ is assumed to have a distribution belonging to an exponential family, i.e. to be distributed according to a discrete or continuous density with the form
\begin{equation}\label{EFdensity}
f(y|\theta, \lambda) = \exp \left\{ \frac{y' \theta - b(\theta)}{\lambda} + c(y,\lambda) \right\}
\end{equation}
where $\theta$ is a $K$-dimensional vector of parameters and $\lambda > 0$ is the dispersion (nuisance) parameter. If we denote with $\mu (\theta)$ and $\Sigma (\theta)$ the expected value $E_{\theta} Y$ and covariance matrix $\mathrm{cov}_{\theta} Y$, respectively, it is well know that 
\begin{eqnarray}
\mu (\theta) &=& \frac{\partial b(\theta)}{\partial \theta},\label{mean} \\
\Sigma (\theta) &=& \lambda \frac{\partial^2  b(\theta)}{\partial \theta \partial \theta'}.\label{covariance}
\end{eqnarray}

In the multivariate generalized linear model the parameters $\theta$ and $\lambda$ are not constrained to be constant for different observations. For the $i$-th observation $y_i$, $i=1,\ldots,n$, a design $K \times p$ matrix $Z_i$ and the $p$-dimensional vector $\beta$ of parameter of interest provide the linear predictor $\eta_i = Z_i \beta$. The expected value $\mu_i=\mu(\theta_i)$ is related to the linear predictor through the response function $h$, $\mu_i = h(\eta_i)$. When the inverse of $h$ exists it is denoted by $g$ and it is called the link function, $g(\mu_i)=\eta_i$. The vector $\theta_i$ can now be expressed as $\theta_i = u(Z_i \beta)=\mu^{-1} h(Z_i \beta)$ and in what follows we will highlight the dependency of the functions on $\beta$ more than on $\theta_i$. 

The loglikelihood for $n$ independent observations, up to a constant, can be written as
\begin{equation}\label{LogLik}
l(\beta) = \sum_{i=1}^n l_i(\beta_i)=\sum_{i=1}^n \frac{y_i' \theta_i - b(\theta_i)}{\lambda_i}
\end{equation}
where $l_i$ indicates the part related to the $i$-th observation. The parameter $\lambda_i$ are very often expressed as $\phi/\omega_i$ where $\phi$ is a constant parameter and $\omega_i$ are known weights. Here also $\phi$ is supposed to be know. The score function $s(\beta)$ can be expressed as the sum of the individual score functions $s_i(\beta) = \partial l_i(\beta) / \partial \beta$ 
\begin{equation}
s_i(\beta) = Z_i' D_i(\beta)\Sigma_i^{-1} (\beta) (y_i - \mu_i(\beta)),
\end{equation}
$D_i(\beta)$ indicating the derivative of $h(\eta)$ evaluated at $\eta_i = Z_i \beta$. It will be convenient to consider also the following form of the individual score functions
\begin{equation}
s_i(\beta) = Z_i' W_i(\beta)D_i^{-1}(\beta) (y_i - \mu_i(\beta))
\end{equation}
with
\begin{displaymath}
W_i(\beta)= D_i(\beta)\Sigma_i^{-1}(\beta)D_i(\beta)'.
\end{displaymath}
If we denote with $H_l(\beta)$ the Hessian matrix of the loglikelihood we can write it as the sum of the individual component $H_{l_i}(\beta)$. Analogously the Fisher information, $F(\beta)=- E_{\beta}H_l(\beta)$ can be expressed as the sum of the individual components $F_i(\beta)$
\begin{equation}
F_i(\beta)=E_\beta{s_i(\beta) s_i(\beta)'} = Z_i' W_i(\beta) Z_i.
\end{equation}
Finally, the individual observed information is related to the expected information by the following relationship: 
\begin{equation}\label{Rel_ExpInf_ObsInf}
F_{i,obs}(\beta)=F_i(\beta) - R_i(\beta)
\end{equation}
where
\begin{equation}\label{partOfObsInf}
R_i(\beta) = \sum_{r=1}^q Z_i' U_{ir}(\beta) Z_i(y_{ir}-\mu_{ir}(\beta)).
\end{equation}
In equation (\ref{partOfObsInf}) $U_{ir}(\beta)=\partial^2 u_r(\eta_i)/\partial \eta \partial \eta'$, where $u_r(\eta_i)$ is the $r$-th component of $u(\eta_i)$ evaluated at $\eta_i = Z_i \beta$.

\subsection*{Numerical methods}

Numerical methods to get maximum likelihood estimates can be compactly presented by using matrix notation. Here we give the details for Newton-Raphson, Fisher scoring and the algorithm outlined in the paper. Let us denote with $Z$ the total design matrix
\begin{displaymath}
Z=\left[ \begin{array}{c}
Z_1 \\
\vdots \\
Z_n
\end{array} \right]
\end{displaymath}
and let us denote with $y'= (y_1',\ldots,y_n')$ and $\mu(\beta)'=(\mu_1(\beta)',\ldots,\mu_n(\beta)')$ the total vector of observations and the respective vector of expected values. Now consider that the following matrices have a block diagonal form
\begin{displaymath}
\Sigma(\beta)=\mathrm{diag}(\Sigma_i(\beta)), \quad W(\beta)=\mathrm{diag}(W_i(\beta)),
\end{displaymath}
\begin{displaymath}
 D(\beta)=\mathrm{diag}(D_i(\beta)).
\end{displaymath}
With this notation it is easy to see that
\begin{displaymath}
s(\beta)=Z'D(\beta)\Sigma^{-1}(\beta)[y - \mu(\beta)], \quad F(\beta)=Z'W(\beta)Z.
\end{displaymath}

The Newton-Raphson and Fisher scoring algorithms are given by the following iterations, respectively:
\begin{eqnarray}
\beta^{(t+1)} &=& \beta^{(t)} - [H_l(\beta^{(t)})]^{-1} s(\beta^{(t)}) \label{Newton-Raphsoniteration}\\
\beta^{(t+1)} &=& \beta^{(t)} - [E_\beta{H_l(\beta^{(t)})}]^{-1} s(\beta^{(t)}).
\end{eqnarray}
When the natural link $g=\mu^{-1}$ is used, Newton-Raphson and Fisher scoring coincide. Fisher scoring can be summarized as iteratively reweighted least squares, IRLS: 
\begin{eqnarray*}
\beta^{(t+1)} &=& \beta^{(t)} + \left(Z'W(\beta^{(t)})Z\right)^{-1} s(\beta^{(t)})\\
 &=& \left(Z'W(\beta^{(t)})Z\right)^{-1} \left[ Z'W(\beta^{(t)})Z     \beta^{(t)} + s(\beta^{(t)}) \right] \\
 &=& \left(Z'W(\beta^{(t)})Z\right)^{-1}Z'W(\beta^{(t)}) \left[Z \beta^{(t)} + D^{-1}(\beta^{(t)})(y - \mu(\beta^{(t)})) \right]
\end{eqnarray*}
where $D^{-1}(\beta)=\mathrm{diag}(D_i^{-1}(\beta))$. In the last equation $\left[Z \beta^{(t)} + D^{-1}(\beta^{(t)})(y - \mu(\beta^{(t)})) \right]$ is usually called the working variate. Using the natural link we have $\theta_i=\eta_i=Z_i \beta$ and the computations are dramatically simplified because the relationships (\ref{mean}) and (\ref{covariance}) give:
\begin{eqnarray*}
D_i(\beta)&=&\frac{\partial h(\eta_i)}{\partial \eta_i}=\frac{\partial \mu_i}{\partial \theta_i}\\
&=& \frac{\partial^2 b(\theta_i)}{\partial \theta_i \partial \theta_i'}=\frac{w_i}{\phi}\Sigma_i(\beta).
\end{eqnarray*}
Therefore score functions and Fisher information simplify to
\begin{displaymath}
s_i(\beta)=Z_i'\left(\frac{w_i}{\phi}\right)\left[y_i - \left. \frac{\partial b(\theta)}{\partial \theta} \right\vert_{\theta=Z_i\beta}\right]
\end{displaymath}
and
\begin{displaymath}
F_i(\beta)=Z_i' \left(\frac{w_i}{\phi}\right)  \left. \frac{\partial^2 b(\theta)}{\partial \theta \partial \theta '}\right\vert_{\theta=Z_i\beta} Z_i
\end{displaymath} 
respectively. It is evident that from the exponential family form in (\ref{EFdensity}) we have all the necessary ingredients for the algorithm. The same is true for the proposed Levenberg-Marquardt approach because for the generalized linear model with natural link our arguments in section \ref{sec:LM} apply \textit{ceteris paribus}, replacing $\theta$ with $\beta$.

\section*{Appendix B}

In this appendix we give a formal proof of convergence for the algorithm discussed in the paper with reference to the generalized linear model with natural link introduced in the previous section. Let $M(\beta)$ be the iteration map
$$
M(\beta)=\beta-\left\{ H_l(\beta) + \gamma P(\beta)  \right\}^{-1} s (\beta)
$$
where $\gamma$ is a positive function and $P(\beta)$ is a symmetric negative definite matrix.
\begin{lemma}\label{lemma_differential}
Let us denote with $\beta_{\infty}$ the maximum likelihood estimate for the parameters of interest. Then the differential of the iteration map is given by:
\begin{displaymath}
 dM(\beta_{\infty})=I - \left\{H_l(\beta_{\infty}) + \gamma P(\beta_{\infty}) \right\}^{-1}H_l(\beta_{\infty}).
\end{displaymath}
\end{lemma}
\begin{proof}
We have $H_l(\beta)=E_{\beta}H_l(\beta)=-F(\beta)$ and the last term is negative definite due to the regularity assumptions. $ \gamma P(\beta)$ is also negative definite. Therefore  $\left\{H_l(\beta) + \gamma P(\beta) \right\}$ is always invertible. Let us denote with $\zeta(\beta)$ the function
\begin{displaymath}
\zeta(\beta)=-\left\{H_l(\beta) + \gamma P(\beta) \right\}^{-1}s (\beta).
\end{displaymath}
We can rearrange the terms as $-\left\{H_l(\beta) + \gamma P(\beta) \right\}\zeta(\beta)= s (\beta)$.  An application of the implicit function theorem in a neighbourhood of $\beta_{\infty}$ gives
\begin{eqnarray*}
d\zeta(\beta)&=&- \left \{ d_{\zeta(\beta)} \left( -\left[ H_l(\beta) + \gamma P(\beta) \right]\zeta(\beta) \right) \right\}^{-1}H_l(\beta)\\
&=&\left \{H_l(\beta) + \gamma P(\beta) \right \}^{-1}H_l(\beta).
\end{eqnarray*}
Therefore $d\zeta(\beta_{\infty}) =\left \{H_l(\beta_{\infty}) + \gamma P(\beta_{\infty}) \right \}^{-1}H_l(\beta_{\infty})$ and $dM(\beta_{\infty})=I - d\zeta(\beta_{\infty})$.
\end{proof}

\begin{theorem}
The algorithm given by equation (\ref{LMiteration}) is locally attracted to the maximum likelihood estimate $\beta_{\infty}$ at a linear rate equal to the spectral radius of 
$$I - \left\{ H_l(\beta_{\infty}) + \gamma P(\beta_{\infty}) \right\}^{-1}  H_l(\beta_{\infty}),$$ 
or at a better rate.
\end{theorem}

\begin{proof}
We remark that $\gamma$ is a positive function. The point of maximum for $l(\beta)$ is a fixed point for $M(\beta)$. According to Proposition 15.3.1 in \citet{KL_NAS} it suffices to show that all eigenvalues of the differential $dM(\beta_{\infty})$ lie on the open interval $(0,1)$. By Lemma \ref{lemma_differential} the following passages hold:
\begin{eqnarray*}
dM(\beta_{\infty}) &=& I - \left\{ H_l(\beta_{\infty}) + \gamma P(\beta_{\infty})  	\right\}^{-1}  H_l(\beta_{\infty})\\
	&=&\left\{ H_l(\beta_{\infty}) + \gamma P(\beta_{\infty})  	\right\}^{-1} \left\{ H_l(\beta_{\infty}) + \gamma P(\beta_{\infty}) - H_l(\beta_{\infty}) 	\right\}\\
	&=& \left\{ H_l(\beta_{\infty}) + \gamma P(\beta_{\infty})  	\right\}^{-1} \left\{ \gamma P(\beta_{\infty}) \right\}.
\end{eqnarray*}
The maximum and minimum eigenvalues of $dM(\beta_{\infty})$ are determined by the maximum and minimum values of the Rayleigh quotient ($v \neq 0$):
\begin{eqnarray*}
R(v)&=&\frac{v'\left[ \gamma P(\beta_{\infty}) \right]v}{v'\left[H_l(\beta_{\infty}) +\gamma P(\beta_{\infty}) \right]v}\\
&=& 1 - \frac{v'H_l(\beta_\infty)v}{v'\left[H_l(\beta_{\infty}) +\gamma P(\beta_{\infty}) \right]v}.
\end{eqnarray*}
The quantities $H_l(\beta_{\infty})$, $\gamma P(\beta_{\infty})$ and $H_l(\beta_{\infty}) +\gamma P(\beta_{\infty})$ are definite negative and therefore $0<R(v)<1$.
\end{proof}

\end{document}